\documentclass[12pt,onecolumn,a4paper]{IEEEtran}
\IEEEoverridecommandlockouts

\usepackage[applemac]{inputenc}
\usepackage[english,american]{babel}
\usepackage{amsmath}
\usepackage{amssymb}
\usepackage{amsthm}
\usepackage{thmtools}
\usepackage{siunitx}
\usepackage{cite} 
\usepackage{tikz}
\usepackage{pgfplots}
\usepackage{graphicx}
\usepackage{comment}
\ifCLASSOPTIONcompsoc
\usepackage[caption=false,font=normalsize,labelfont=sf,textfont=sf]{subfig}
\else
\usepackage[caption=false,font=footnotesize]{subfig}
\fi

\usepackage{ctable}

\newtheorem{Lemma}{Lemma}
\newtheorem{Corollary}{Corollary}
\newtheorem{Proposition}{Proposition}

\newcommand{\Sv}{\mathbf{S}}
\newcommand{\onev}{\mathbf{1}}

\newcommand{\tbs}{t_\mathrm{bs}}
\newcommand{\tn}{t_\mathrm{d}}

\newcommand{\TDW}{\overline{T}_\mathrm{dw}}
\newcommand{\TDD}{\overline{T}_{\eta }}
\newcommand{\TDDs}{\overline{T}_1 }
\newcommand{\TDDr}{\overline{T}_0  }
\newcommand{\etas}{\eta_1 }
\newcommand{\etar}{\eta_0 }
\newcommand{\TDDi}{\overline{T}_i  }
\newcommand{\etai}{\eta_i }

\newcommand{\ps}{p_\mathrm{s}}

\newcommand{\nc}{n_\mathrm{c}}
\newcommand{\Mc}{M_\mathrm{c}}
\newcommand{\Ta}{T_\mathrm{a}}
\newcommand{\Ts}{T_\mathrm{s}}

\IEEEoverridecommandlockouts

\title{MDS-Coded Distributed Caching for\\ Low Delay Wireless Content Delivery}
\author{Amina Piemontese,~\IEEEmembership{Member,~IEEE}, and Alexandre Graell i Amat,~\IEEEmembership{Senior Member,~IEEE}
\thanks{The authors are with the Department of Signals and Systems, Chalmers University of Technology, 412 96 Gothenburg, Sweden (e-mail: \{aminap,alexandre.graell\}@chalmers.se.}
\thanks{Amina Piemontese is supported by a Marie Curie fellowship (contract 658785-DISC-H2020-MSCA-IF-2014). This work was also was partially funded by the Swedish Research Council under grant \#2011-5961.}
\thanks{The paper was presented in part at the International Symposium on Turbo Codes \& Iterative Information Processing, Brest, France, Sep. 2016}} 
\begin{document}

\maketitle

\begin{abstract}
We investigate the use of maximum distance separable (MDS) codes to cache popular content to reduce the download delay of wireless content delivery. In particular, we consider a cellular system where devices roam in an out of a cell according to a Poisson random process. Popular content is cached in a limited number of the mobile devices using an MDS code and can be downloaded from the mobile devices using device-to-device communication. We derive an analytical expression for the delay incurred in downloading content from the wireless network and show that distributed caching using MDS codes can dramatically reduce the download delay with respect to the scenario where content is always downloaded from the base station and to the case of uncoded distributed caching.
\end{abstract}


\section{Introduction}
The proliferation of mobile devices and the surge of a myriad of multimedia applications has resulted in an exponential growth of the mobile data traffic. In this context, 
wireless caching has emerged as a powerful technique to overcome the backhaul bottleneck, by reducing the backhaul rate and the delay in retrieving content from the network. The key idea is to store popular content closer to the end users. In \cite{Shan13}, a novel system architecture named \emph{femtocaching} was proposed. It consists of deploying a number of small base stations (BSs) with large storage capacity, in which content is stored during periods of offpeak traffic. The mobile users can download  content from the small BSs, which results in a higher throughput per user. In \cite{Gol14}, it was proposed to store content directly in the mobile devices.  Users can then retrieve content from neighboring devices using device-to-device (D2D) communication or, alternatively, from the serving BS. 

In both scenarios, content may be stored using an erasure correcting code, which brings gains with respect to uncoded caching \cite{Paa13,Ped15,Bio15,Ped16}. The use of erasure correcting codes establishes an interesting link between distributed caching for content delivery and distributed storage (DS) for reliable data storage. The key difference is that in the wireless network scenario, data can be downloaded from the storage nodes (the small BSs or the mobile devices) but also from a serving macro BS, which has always the content available. Therefore, the reliability requirements in DS for reliable data storage can be relaxed. In \cite{Bio15}, the placement of content encoded using a maximum distance separable (MDS) code to small BSs was investigated and it was shown that a careful placement allows to significantly reduce the backhaul rate. In \cite{Paa13}, for the scenario where content is stored directly in the mobile devices, the repairing of the lost data when a device storing data leaves the network was considered. Assuming  instantaneous repair, the communication cost of data download and repair was investigated. In \cite{Ped15,Ped16}, a repair scheduling where repair is performed periodically was introduced and analytical expressions for the overall communication cost of content download and data repair as a function of the repair interval were derived. Using these expressions, the overall communication cost entailed by storing content using MDS codes, regenerating codes \cite{Dim10}, and locally repairable codes \cite{Pap14} was evaluated in \cite{Ped16} and it was shown that storing content using erasure correcting code can reduce the overall communication cost with respect to the scenario where content is downloaded solely from the BS. 

In this paper, we consider a similar cellular network scenario as the one in \cite{Paa13,Ped16}, where content is stored in a number of mobile devices using an erasure correcting code. Mobile devices roam in an out of a cell according to a Poisson random process. However, as opposed to \cite{Paa13,Ped16}, where the download of a single file is considered, here we consider 
that users may request files, of different popularity, from a library of files. Our focus is on the delay of retrieving content from the network, which was not considered in \cite{Paa13,Ped16}. We derive analytical expressions for the download delay if content is stored in the mobile devices using MDS codes and show that MDS-coded distributed caching can significantly reduce the download delay with respect to the case where content is solely downloaded from the BS and the case where uncoded caching is used. The download delay of a single file was analyzed in~\cite{PiGr16}.

The remainder of the paper is organized as follows. The system model is introduced in Section~\ref{s:system_model}. The average download delay incurred when MDS-coded distributed caching is used is analyzed in Sections~\ref{s:download} and Section~\ref{s:D2D}. Section~\ref{sec:NumericalResults} presents and discusses numerical results and finally some conclusions are drawn in Section~\ref{sec:Conclusions}.

\textit{Notation.} The probability density function (pdf) of a random variable $X$ is denoted by $f_X(\cdot)$ and the  expectation with respect to $X$ is denoted by $\mathbb{E}_X\{\cdot\}$. Probability is denoted by $\Pr\{\cdot\}$ and $\onev_{i}$ represents the all-ones vector of length $i$. We denote by $\pi_m(\rho)$ the stationary distribution of an $\mathsf{M}/\mathsf{M}/\infty$ queueing system described by a Poisson birth-death process with arrival rate $\alpha$ and departure rate per node $\delta$, which is given by
\begin{equation}\nonumber
\pi_m(\rho)=\frac{{\rho}^m}{m !} e^{-\rho}\, ,
\end{equation}
where  $\rho=\alpha/\delta$.

\section{System Model}\label{s:system_model}

We consider a single cell in a cellular network where $M$ mobile devices, referred to as nodes,  request files, each of size $B$ bits, from a library of $Z$ files. The files have different popularities and accordingly have a given probability to be requested.  Depending on the placement strategy, some files are encoded and stored into $n\leq M$ mobile devices,  referred to as storage nodes. For ease of language, the set of storage nodes is referred to as the DS network and nodes not storing any content are referred to as \emph{regular} nodes. A copy of each encoded file is also available at the BS serving the cell.  A node requesting a file attempts to retrieve it from the storage nodes using D2D communication, and, if the file cannot be completely retrieved from the DS network, the BS assists in providing the missing data. In order to increase the system efficiency, we allow multiple D2D communications to coexist if they are sufficiently far apart in space. Therefore, we divide the cell in $C$ virtual clusters  and assume that the size of the cluster and the transmit power are properly chosen such that only one D2D communication can be established between any two nodes in the cluster and the interference across different clusters can be neglected. A similar model is considered in~\cite{Gol14,Ji16}.

\textit{Data allocation and coding strategy.}
We adopt a deterministic allocation strategy, where the $F\leq Z$ most popular files are cached in a distributed fashion in $n$ storage nodes in the cell, according to the storage capacity of the devices. These files are partitioned into $k$ packets, called symbols, of $B/k$ bits each and are encoded into $n$ coded symbols using an $(n,k)$ MDS erasure correcting code of rate $r = k/n$. We use the same code for every file in order to simplify the analysis. 
We assume that each storage node stores a single symbol for each of the $F$ most popular files. Overall, $nF$ symbols are stored in $n$ storage nodes and no two storage nodes store the same symbol.
We model the popularity of the files in the library using the time-invariant Zipf distribution~\cite{Bre99}.\footnote{The popularity of the files in mobile data traffic does not change very rapidly, i.e., it can be considered constant during the day.} Accordingly, the probability that the $i$th file is requested is
\begin{equation}\label{e:Zipf}
z_i=\frac{1/i^{\sigma}}{\sum_{j=1}^Z 1/j^{\sigma}   } \, , \quad \qquad  1\leq i \leq Z\, ,
\end{equation}
where parameter $\sigma$ regulates the relative popularity of the files. In the following, the set of $F$ files stored in the cache of the mobile devices will be referred to as the DS library.

We assume that the mobile devices are free to move inside the cell. We consider a uniform spatial distribution of the nodes in the cell, and hence there are $\Mc=M/C$ devices per cluster on average and among them $\nc=n/C$ storage nodes. 
We focus on a single cluster in isolation, and  assume that the devices roam in and out of it. The arrival, departure and request model of the nodes are borrowed from \cite{Ped16}. The considered scenario is shown in Fig.~\ref{f:cluster}. 

\textit{Arrival-departure model.}
We assume that nodes arrive to the cluster according to a Poisson random process with exponential independent, identically distributed (i.i.d.) random inter-arrival times $\Ta$ with pdf
\begin{equation}\label{e:arrival}
f_{\Ta}(t)=\Mc\lambda e^{-\Mc\lambda t}, \qquad \lambda\geq 0, t\geq 0,
\end{equation}
where $\Mc\lambda$ is the expected arrival rate and $t$ is time, measured in time units (t.u.). The nodes stay in the cluster for an i.i.d. exponential random lifetime $T_ \ell$ with pdf
\begin{equation}\label{e:departure}
f_{T_\ell}(t)=\mu e^{-\mu t}, \qquad \mu\geq 0, t\geq 0,
\end{equation}
where $\mu$ is the expected departure rate per node. We assume that $\mu=\lambda$, which implies that the expected number of nodes in the cluster is $\Mc$. This model corresponds to an $\mathsf{M}/\mathsf{M}/\infty$ queuing model and the probability that there are $i$ nodes in the cluster is $\pi_i(\Mc)$.
The arrival of storage nodes to the cluster can also be described as a Poisson random process. In particular, the inter-arrival times $\Ts$ of the set of storage nodes has pdf
\begin{equation}\nonumber
f_{\Ta}(t)=\nc\lambda e^{-\nc\lambda t}, \qquad \lambda\geq 0, t\geq 0\,.
\end{equation}
The related lifetime is described by~(\ref{e:departure}) and the probability that there are $i$ storage nodes in the cluster is $\pi_i(\nc)$.\footnote{The Poisson model is largely used in the case of uniform mobility and its popularity is also due to its tractability. However, we would like to remark that while it is able to capture the mobility in one cluster, this model does not guarantee that the total number of storage nodes in the cell is constant and equal to $n$. More precisely, the only guarantee is that there are on average $\nc$ storage nodes per cluster, but there are no constraints on their instantaneous number, which can even exceed $n$. On the other hand, the probability of having a high number of storage devices in one cluster is generally very low. For $\nc=9$, we have $\pi_{18}( \nc)=3\cdot 10^{-3}$, $\pi_{27}(\nc)=6.6\cdot 10^{-7}$ and $\pi_{91}( \nc)=4\cdot 10^{-48}$. The same consideration holds for the total number of mobile devices.}
\begin{figure}[!t]
\centering{} \includegraphics[width=88mm]{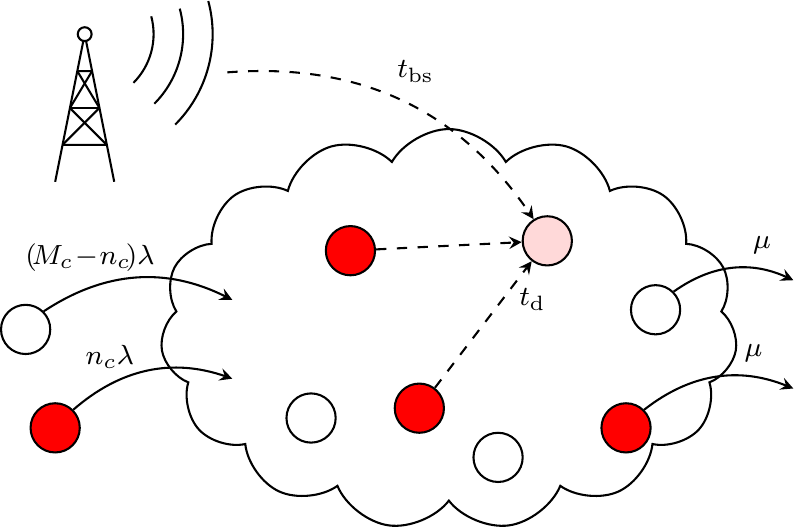}
\caption{An example of cluster where nodes roam in and out according to a Poisson random process: we have on average $\Mc$ mobile devices, and $\nc$ storage nodes among them (red circles), caching one different coded symbol for each of the most popular files. A device requesting a file (pink circle), must collect $k$ symbols. It attempts to recover them by using the DS network if the requesting file is stored in the devices. It uses the BS to collect the symbols that it is not able to download from the devices. The download of a symbol from a storage node takes $\tn$ t.u., and from the BS $\tbs$ t.u..}\label{f:cluster}
\end{figure}

\textit{DS network update.}
We assume that the nodes storing content that arrive to the cluster from neighboring clusters are not immediately available for download, but the BS serving the cell keeps track of them and periodically updates and broadcasts to all mobile devices the list of storage nodes in the cell every $\Delta$ t.u..
In the sequel, parameter $\Delta$ is referred to as the update interval and the set of storage nodes in the list broadcasted by the BS as the DS list.

\textit{Data delivery.}
Nodes request the file at random times with i.i.d. random inter-request time $T_r$ with pdf
\begin{equation}\label{e:request time}
f_{T_r}(t)=\omega e^{-\omega t}, \qquad \omega\geq 0, t\geq 0,
\end{equation}
where $\omega$ is the expected request rate per node. We focus on the download process. The node that requests a file attempts to retrieve it from the DS network using D2D communication. Thanks to the MDS property, an encoded file can be reconstructed by accessing any $k$ encoded symbols. If the file cannot be completely retrieved from the DS network, the BS assists in providing the missing coded symbols. The download of a coded symbol from a storage node incurs $\tn$ t.u. and from the BS $\tbs$ t.u.. We assume that $\tbs\gg \tn$ due to the congestion of the BS-to-node link and the fact that D2D communication occurs over a better channel due to the reduced distance between the involved nodes. We further assume that only one D2D link at a time can be established, and that the D2D communication does not interfere with the communication between the BS and the nodes. We say that the D2D network is \emph{idle} if there is no active D2D communication in the cluster. If the D2D network is not idle when one node requests the file, the whole file is downloaded from the BS. Moreover, to simplify the analysis, we assume that multiple BS-to-node links can coexist.


\section{File Average Download Delay}\label{s:download}

We investigate the average time that is required to retrieve one file from the wireless network, referred to as the download delay. If a requested file is stored in the DS library, the requesting node attempts to retrieve it from the DS network using D2D communication, otherwise the file is entirely downloaded from the BS. Therefore, we introduce the binary random variable (RV) $H\in\{0,1\}$  which describes the hitting of the DS cache, i.e., $H=1$ if a file of the DS library is requested and $H=0$ otherwise. Moreover, the D2D network can be used only if it is idle, i.e., if there are no active D2D communications. Accordingly, we introduce the binary RV $I\in\{0,1\}$ that describes the status of the D2D network. $I=1$ if the network is idle and $I=0$ otherwise. If the D2D network is idle, the requesting node tries to collect the necessary coded symbols from the nodes of the DS list provided by the BS using D2D communication. If the requesting node is a storage node of the DS list, it needs to download $k-1$ symbols, otherwise $k$ symbols must be downloaded.  We thus introduce the binary RV $R\in\{0,1\}$, which represents the type of request, i.e., $R=1$ for requests originating from a node that belongs to the DS list and $R=0$ for the other requests. The download from the storage nodes can be either fully successful or partially accomplished, in which case the requesting node turns to the BS to recover the missing symbols. 
On the other hand, if the D2D network is not idle and the requested file is stored in the DS library, the node downloads $k$ or $k-1$ symbols from the BS, depending on the type of node.

From the discussion above, the average file download delay, $\TDW$, may be formalized as

\begin{Proposition} The average file download for the cellular network described in Section~\ref{s:system_model} where the $F$ most popular files are stored in the mobile devices using an $(n,k)$ MDS code is
\begin{align}
\TDW= & \Pr\{H=0\} k  \tbs  + \Pr\{I=1\} \Pr\{H=1\} \Big(\TDD + (k-\Pr\{R=1\}-\eta) \tbs \Big) \nonumber\\
&+ \Pr\{I=0\} \Pr\{H=1\} (k-\Pr\{R=1\}) \tbs \label{e:T} \, ,
\end{align}
where $\eta$ is the average number of coded symbols downloaded per request using D2D communication and $\TDD$, referred to as the average D2D download delay, is the corresponding delay.
\end{Proposition}

The computation of $\eta$, $\TDD$ and $\Pr\{R=1\}$ is addressed in Section~\ref{s:D2D}. 
The probability of hitting the cache can be expressed as
\begin{equation}\nonumber
\Pr\{ H=1\} =\sum_{i=1}^F z_i\, ,
\end{equation}
where the probabilities $z_i$ are given in~(\ref{e:Zipf}). It follows that $\Pr\{ H=1\} =1$ if $F=Z$.

The next step is the computation of the probability that the D2D network is idle.  Let $I^{(\ell)}$ be the status of the network at the time of the $\ell$th request. It follows
\begin{equation}\label{e:idle}
\Pr\{ I=1\} =\lim_{L \to \infty} \frac{1}{L} \sum_{\ell=1}^L \Pr\{ I^{(\ell)}=1\}.
\end{equation}
In order to compute $\Pr\{ I^{(\ell)}=1\}$, we introduce the RV $W^{(j)}$ that denotes the time instant of the $j$th request. Also, let $T^{(j)}$ be the time during which the D2D network is occupied by the $j$th request. The D2D network is idle at the time of the  $\ell$th request if none of the previous requests is still using D2D communication. Therefore, $\Pr\{ I^{(1)}=1\}=1$ and 
\begin{equation}\label{e:P_I exact}
\Pr\{ I^{(\ell)}=1\}\!=\!\prod_{i <\ell} \! \Pr\{  W^{(\ell)} \!\!>  \!\!W^{(\ell-i)} \!+ T^{(\ell-i)}  \}, \,\ell>1 .
\end{equation}
Assuming that if the D2D network is not idle at time $W^{(\ell)} $ is because of the $(\ell-1)$th request, the product in (\ref{e:P_I exact}) reduces to the term involving the $(\ell-1)$th request only, i.e.,
\begin{align}\label{e:P_I approx1}
\Pr\{ I^{(\ell)}=1\}&\simeq \Pr\{  W^{(\ell)} >  W^{(\ell-1)} + T^{(\ell-1)}  \} \\
&= \int_0^\infty \Pr\{  W^{(\ell)} >  W^{(\ell-1)} + t  \} f_{T^{(\ell-1)} }(t) d t\, .\nonumber
\end{align} 
Since the requests are i.i.d. with inter-request time distributed as in (\ref{e:request time}) and on average there are $\Mc$ nodes in the cluster,  we can compute
\begin{equation}\nonumber
\Pr\{  W^{(\ell)} >  W^{(\ell-1)}+ t  \} = e^{-\omega \Mc t}\, ,\quad t\geqslant 0,\quad\ell>1\, ,
\end{equation}
and (\ref{e:P_I approx1}) can be written as 
\begin{equation}\nonumber
\Pr\{ I^{(\ell)}=1\}\simeq \mathbb{E}_{T^{(\ell-1)}} \{  e^{-\omega \Mc T^{(\ell-1)}}  \} ,\quad\ell>1,
\end{equation}
If $\omega T^{(\ell-1)} \ll 1$,
\begin{equation}\label{e:P_I approx2}
e^{-\omega \Mc T^{(\ell-1)}} \simeq 1-\omega \Mc T^{(\ell-1)}
\end{equation}
and
\begin{align}
\Pr\{ I^{(\ell)}=1\} & \simeq  \mathbb{E}_{T^{(\ell-1)}} \{  e^{-\omega \Mc T^{(\ell-1)}}  \}    \nonumber\\
                           & \simeq  \mathbb{E}_{T^{(\ell-1)}} \{   1-  \omega \Mc T^{(\ell-1)}   \} \nonumber\\
                           & = 1-\omega \Mc \Pr\{ I^{(\ell-1)}=1\} \Pr\{ H=1\}  \TDD .\label{e:idle_ell}
\end{align}
In~(\ref{e:idle_ell}), we used the fact that  the probability of hitting the cache and the average D2D download delay are independent of the request index (if $\ell$ is sufficiently large), as it is proven in Lemma~\ref{l:Tdd_ell} in Section~\ref{s:D2D}. Substituting (\ref{e:idle_ell}) in (\ref{e:idle}) and after some simple calculations, we obtain
\begin{equation}\label{e:approxI}
\Pr\{ I=1\} \simeq \frac{1}{1+\omega \Mc   \Pr\{ H=1\}  \TDD} \, .
\end{equation} 


\section{Download From Storage Nodes}\label{s:D2D}

In this section, we consider the computation of the average D2D download delay $\TDD$ and the average number of coded symbols $\eta$ downloaded per request using D2D communication. We assume that a node cannot download in parallel from multiple nodes, but it serially tries to download the coded file symbols from the nodes in the DS list. When a node requests the file, if the D2D network is idle and the requested file belongs to the DS library, it randomly chooses one of the storage nodes from the list supplied by the BS. After each downloaded symbol, the requesting node randomly chooses another storage node from the DS list and still alive.\footnote{The requesting node uses the storage nodes alive at the moment of its request even if, during the download process, new storage nodes are included in the DS list after the periodic restoration.} 
We assume that a requesting node that has collected fewer than the $k$ symbols necessary to reconstruct the file turns to the BS when all the reference storage nodes  left or when the download of a symbol fails, even if other storage nodes are available. To simplify the analysis, we assume that both cases (the failed symbol download and the absence of storage nodes) incur $\tn$ t.u., even if the node could contact the BS earlier. We also assume that the download from the D2D network fails if the requesting node itself leaves the cluster before collecting $k$ symbols. In this case, the download is also completed from the BS. 

To derive the average D2D download delay, we introduce three RVs describing the number of nodes of different type that are present in the cluster at the instant of a request: the number of storage nodes of the DS list, the total number of storage nodes (belonging or not to the list, the latter corresponding to the storage nodes that arrive to the cluster after the DS list update and that have not left the cluster at the time of the request), and the number of regular nodes. In particular, we denote by $X_1\in\{0,\dots,\infty\}$ the RV that describes the number of storage nodes of the DS list when a request arrives. We describe by the RVs $Q\in\{0,\dots,\infty\}$ and $V\in\{0,\dots,\infty\}$ the total number of storage nodes and the number of regular nodes at the instant of a request, respectively. Moreover, we denote by $Y\in\{0,\dots,\infty\}$  the RV that represents the total number of storage nodes (belonging or not to the DS list) at the beginning of the update interval of length $\Delta$. In the following three lemmas, we give a probabilistic description of the above RVs. 

\begin{Lemma}\label{l:px1}
The probability that there are $x\geq0$ storage nodes of the DS list at the time of a request is
\begin{equation}\label{e:px1}
\Pr\{X_1=x \}= \frac{\sum_{y=0}^\infty  \pi_y(\nc) \sum_{m=y}^\infty (1-e^{-m\omega \Delta}) \pi_{m-y}(\Mc-\nc)  \Pr\{X_1=x|Y=y\} }{  \sum_{m=1}^\infty (1-e^{-m\omega \Delta}) \pi_m(\Mc)  }\, ,
\end{equation}
where $\Pr\{X_1=x|Y=y\}$ is the probability that $ X_1 $ is equal to $x$, given that $y\geq0$ storage nodes are in the cluster at the beginning of the update interval of length $\Delta$, and is 
\begin{equation}
\Pr\{X_1=x|Y=y\}=\frac{1}{\Delta} \sum_{i'=x}^y \frac{1-p_{i'}}{\mu_{i'}} \prod_{{\substack{ j=x \\ j\neq i' }}  }^y \frac{j}{j-i'}- \frac{1}{\Delta} \sum_{i'=x + 1}^y \frac{1-p_{i'}}{\mu_{i'}}  \prod_{{\substack{ j=x + 1 \\ j\neq i' }} }^y \frac{j}{j-i'}\, ,\label{e:px_J}
\end{equation}
where $\mu_{i'}=i'\mu$ and $p_{i'}=e^{-\mu_{i'}\Delta}$.
\end{Lemma}
\begin{proof}
The proof is given in Appendix~\ref{app:ProofLem1}.
\end{proof}


\begin{Lemma}\label{l:pq}
The probability that there are $q\geq0$ storage nodes in the cluster at the time of a request is given by
\begin{equation}\label{e:pq}
\Pr\{Q=q \}= \frac{\sum_{m=q}^\infty   (1-e^{-m\omega \Delta}) \pi_{m-q}(\Mc-\nc)  }{  \sum_{m=1}^\infty (1-e^{-m\omega \Delta}) \pi_m(\Mc)  } \pi_q(\nc) \, .
\end{equation}
\end{Lemma}
\begin{proof}
The proof follows the same lines as the proof of Lemma~\ref{l:px1}.
\end{proof}

\begin{Lemma}\label{l:pv}
The probability that there are $v\geq0$ regular nodes in the cluster at the time of a request is given by
\begin{equation}
\Pr\{V=v \}= \frac{\sum_{m=v}^\infty   (1-e^{-m\omega \Delta}) \pi_{m-v}(\nc)  }{  \sum_{m=1}^\infty (1-e^{-m\omega \Delta}) \pi_m(\Mc)  } \pi_v(\Mc-\nc) \, .
\end{equation}
\end{Lemma}
\begin{proof}
The proof follows the same lines as the proof of Lemma~\ref{l:px1}.
\end{proof}

Based on the above lemmas, we can compute the probability that the request originates from a storage node of the DS list and the probability of having a given number of storage nodes in the DS list at the time of the request conditioned to the type of request.

Using Bayes' rule, the probability that there are $x\geq 0$ storage node of the DS list alive at the time of a request, conditioned to the type of request, is given by
\begin{equation}\label{e:px1cond}
\Pr\{X_1=x | R=i \}= \Pr\{ R=i|  X_1=x \} \frac{ \Pr\{ X_1=x \}} {  \Pr\{ R=i\} }\,,  i=0,1\, .
\end{equation}
The probability $\Pr\{ X_1=x \}$ is given in Lemma~\ref{l:px1}. We now compute $\Pr\{ R|  X_1 \}$ and $\Pr\{ R \}$. We start with the probability of having one request from the DS list conditioned to the number of storage nodes in the DS list at the time of the request. For $x>0$, it can be written as
\begin{equation}
 \Pr\{ R=1|  X_1=x \} = \sum_{v=0}^\infty \sum_{q=x}^\infty \frac{x}{q+v}\Pr\{V=v|Q=q,X=x\} \Pr\{Q=q|X=x\} \,.
\end{equation}
Clearly, the condition $X=0$ implies that the request cannot originate from a storage node of the DS list, therefore $\Pr\{  R=1 | X_1=0\}  = 0$. We approximate the probability that there are $q$ storage nodes at the instant of the request, given the number of nodes of the DS list alive, by using the steady state probability of a Poisson birth-death process with arrival rate $\lambda (\mathbb{E}_{Q}(Q)-\mathbb{E}_{X_1}(X_1))$ and departure rate $\lambda$. In particular, we compute
\begin{equation}\label{e:pq_approx}
\Pr\{Q=q|X=x \} \simeq \pi_{q-x}(\mathbb{E}_{Q}(Q)-\mathbb{E}_{X_1}(X_1))\,
\end{equation}
where the expectations $\mathbb{E}_{Q}(Q)$ and $\mathbb{E}_{X_1}(X_1)$ are obtained starting from the probabilities~(\ref{e:pq}) and (\ref{e:px1}), respectively.  The number of regular nodes $V$ is independent of the number of storage nodes at the instant of the request, therefore we finally have
\begin{equation}
 \Pr\{ R=1|  X_1=x \} =\sum_{q=x}^\infty  \sum_{v=0}^\infty \frac{x}{q+v}\Pr\{V=v\} \pi_{q-x}(\mathbb{E}_{Q}(Q)-\mathbb{E}_{X_1}(X_1)) \,,
\end{equation}
where $\Pr\{V=v\}$ is given in Lemma~\ref{l:pv}. Note that in the expression above, with some abuse of notation, we used equal sign to avoid carrying all the way the approximation sign due to the approximation introduced in (\ref{e:pq_approx}). 
Starting from this result, we compute the probability that the request originates from the DS list as
\begin{equation}
\Pr\{R=1\}=\sum_{x=0}^\infty \sum_{q=x}^\infty \sum_{v=0}^\infty  \frac{x}{q+v}\Pr\{V=v\} \pi_{q-x}(\mathbb{E}_{Q}(Q)-\mathbb{E}_{X_1}(X_1))  \Pr\{X_1=x \}\, ,
\end{equation}
where $\Pr\{X_1=x\}$ is given in Lemma~\ref{l:px1}. The probability $\Pr\{ R=0|  X_1=x \}$ is easily computed as $1-\Pr\{ R=1|  X_1=x \}$. Similarly, we have $\Pr\{R=0\}=1-\Pr\{R=1\}$.
Following the same approach for the proof of Lemma~\ref{l:px1}, it is easy to show that $\Pr\{X_1^{(\ell)}=x|R^{(\ell)}=i  \}$ and $\Pr\{R^{(\ell)}=i\}$ are independent of the specific request (when $\ell$ grows large), where $R^{(\ell)}$ is the binary RV describing the type of the $\ell$th request.

The case $\Delta=0$ represents the case of instantaneous update, where the nodes contact directly the BS when they request a file and receive the list of the storage nodes through a dedicated link. For instantaneous update, the number of	 storage nodes at the instant of the request and the type of request is described by the following probabilities
\begin{align}
&\Pr\{X=x\}=\pi_x(\nc)\, ,\\
&\Pr\{R=1\}=\frac{\nc}{\Mc}\, ,\\
&\Pr\{R=1|X=x\}=\sum_{m=x}^\infty \frac{x}{m}\pi_{m-x}(\Mc-\nc)\, .
\end{align}
The probability that there are $x\geq0$ storage nodes at the time of the request given the type of request can be computed by replacing the above probabilities in~(\ref{e:px1cond}).


In order to describe the D2D download, let $S_1$ be the binary RV that describes the success of the download at the first attempt. More precisely, $S_1=1$ represents the successful download of the coded symbol from the first contacted storage node. If the download is not successful from the first contacted storage node, $S_1=0$. Similarly, we denote by $S_j$ the binary RV describing the download at the $j$th attempt and we denote by $\Sv_{[i]}$, $i\ge 1$ the random vector  ($S_1,..., S_i$). In the following, in Lemmas~\ref{l:S1}, \ref{l:Sk}, and \ref{l:Sj}, we derive the probability that no symbols can be downloaded from the D2D network, $ \Pr \{S_1=0| R\}$, the probability that the content is fully recovered from the DS network, $\Pr\{ \Sv_{[k-i]}=\onev_{k-i} | R\} $, and the probability that it is only partially recovered, $\Pr \{\Sv_{[j]}=\onev_{j}, S_{j+1}=0 | R\}$, respectively.

\begin{Lemma}\label{l:S1}
The probability that no symbols are downloaded through D2D communication, conditioned to the type of request, is given by
\begin{equation}\nonumber
\Pr\{ S_1=0|R=i\} =  1 + e^{-\mu \tn}  \Big(  \Pr\{X_1=i|R=i\} + \sum_{g =1}^\infty \sum_{d =0 }^{g} \frac{d}{g}    \Pr\{X_1=g+i|R=i\} \theta(d,g) - 1\Big), i=0,1.
\end{equation}
where $\Pr\{X_1|R\} $ is given in~(\ref{e:px1cond}) and 
\begin{equation}\label{e:theta}
\theta(d,g)=\sum_{i'=g-d}^x e^{-\mu_{i'} \tn} \prod_{{\substack{ j=g-d \\ j\neq i' }} }^x \frac{j}{j-i'}- \sum_{i'=g-d+1}^x   e^{-\mu_{i'} \tn}  \prod_{{\substack{ j=g-d+1 \\ j\neq i' }} }^x \frac{j}{j-i'}\,, \quad d\geq 0, g\geq 0
\end{equation}
with  $\mu_{i'}=i'\mu$. 
 \end{Lemma}
 \begin{proof}
The proof is given in Appendix~\ref{app:ProofLem4}.
 \end{proof}

\begin{Lemma}\label{l:Sk}
The probability that the file can be completely retrieved from the DS network, i.e., the probability that k symbols are downloaded through D2D communication when $R=0$, or $k-1$ when $R=1$, conditioned to the type of request, is given by
\begin{equation}\nonumber
\Pr\{ \Sv_{[k-i]} =\onev_{k-i} |R=i\} =  e^{-(k-i)\mu \tn}   \sum_{g =1}^\infty \sum_{d=0 }^{g} \frac{g-d}{g} \gamma_{k-i}(g,d,i)  , i=0,1,
\end{equation}
where $ \gamma_{j}(g,d,i)  $ is defined by the recursion
\begin{equation}\label{e:gamma}
\gamma_j(g,d,i)=  \theta(d,g) \sum_{g'=1}^\infty \sum_{d'=0}^{g'-1}   \frac{g'-d'}{g'}  E(g,g',d') \gamma_{j-1}\!(g',d',i)
\end{equation}
for $d,g\geq 0$ and $i=0,1$, with  initial condition
\begin{equation}\label{e:gammaIC}
 \gamma_1(g,d,i)=\Pr\{X_1=g+i|R=i\} \theta(d,g)\,,
\end{equation}
and where
\begin{equation}\nonumber
E(g,g',d')= 
\begin{cases}
1 & \text{if} \quad g=g'-d'-1 \\
0 & \text{otherwise}\, .
\end{cases}
\end{equation}
The function $\theta(d,g)$ is given in~(\ref{e:theta}), and $\Pr\{X|R\}$ is given in~(\ref{e:px1cond}).
 \end{Lemma}
\begin{proof}
The proof is given in Appendix~\ref{app:ProofLem5}.
\end{proof}

\begin{Lemma}\label{l:Sj}
The probability of consecutively download $j\geq 1$ symbols and to fail the download of the $j+1$th one is
\begin{equation}
\Pr \{\Sv_{[j]}=\onev_{j}, S_{j+1}=0 |R=i \} =  \gamma_{j+1}(0,0,i)a_{j+1} +  \sum_{g=1}^\infty \sum_{d=0}^d \Big(  \frac{d}{g} \gamma_{j+1}(g,d,i)a_{j+1} + \frac{g-d}{g} \gamma_{j}(g,d,i)b_{j+1}  \Big)   \nonumber\, ,
\end{equation}
where $a_j=e^{-j\mu\tn}$, $b_j=e^{-(j-1)\mu\tn}(1-e^{-\mu \tn})$, and $\gamma_{j}(g,d,i)$ is given in~(\ref{e:gamma}).
 \end{Lemma}
\begin{proof}
The proof follows the same lines as the proof of Lemma~\ref{l:Sk}.
\end{proof}

Finally, the average D2D download delay and the average number of downloaded symbols from DS network are given in the following theorem.

\newtheorem{Theorem}{Theorem}
\begin{Theorem}\label{t:T_NI}
Consider the network described in Section~\ref{s:system_model}, where an $(n,k)$ MDS erasure correcting code is employed and where there are $\nc$ storage nodes in the cluster on average. Let $\tn$ be the time to download a symbol through D2D communication. The average D2D download delay and the corresponding average number of downloaded symbols are given by
\begin{align}
\TDD =& \TDDs \ps + \TDDr (1-\ps) \nonumber \\  
\eta= & \etas \ps + \etar (1-\ps) \nonumber
\end{align}
where $\ps=\Pr\{R=1\}$ is the probability that the request comes from a storage node of the DS list, and
\begin{align}
\etai= &(k-i) \Pr\{ \Sv_{[k-i]}=\onev_{k-i} | R=i\}  + \sum_{j=1}^{k-1-i} j   \Pr \{\Sv_{[j]}=\onev_{j}, S_{j+1}=0| R=i \} \, ,\nonumber\\
\TDDi= & \tn \Big( \etai +  c_{k,i}\Pr \{\ S_{1}=0 | R=i\} + \sum_{j=0}^{k-1-i}   \Pr \{\Sv_{[j]}=\onev_{j}, S_{j+1}=0 | R=i\} \Big)\, , \qquad \qquad i=0,1 \nonumber  \end{align}
where $ c_{k,i}=1$ for $k-i>0$ and $ c_{k,i}=0$ otherwise. 
\begin{proof}
The average D2D download delay is obtained as the sum of the average D2D delays in the case of requests originated from the DS list and of requests originating from the other nodes, weighted by the probabilities $p_s$ and $1-p_s$, respectively. The same approach is used for the corresponding average number of downloaded symbols. According to our model, the requesting node completes the download of $k-i$ symbols from the DS network in $(k-i) \tn$ t.u. with probability $\Pr\{ \Sv_{[k-i]}=\onev_k  | R=i \} $, while the partial download of $j<k-i$ symbols happens with probability \mbox{$\Pr \{\Sv_{[j]}=\onev_{j}, S_{j+1}=0  | R=i\}$} and incurs $(j+1)\tn$ t.u.. For $k-i>0$, in the computation of the average D2D download delay, we also consider the case where download from the DS network completely fails. The corresponding probability is $ \Pr \{S_1=0 | R=i\}$ and the delay is $\tn$. When the request originates from the DS list and $k=1$, i.e. $k=i$, no symbols need to be downloaded, therefore $\TDDs$ and $\etas$ are equal to zero.
\end{proof}
\end{Theorem}
\begin{Corollary}\label{l:Tdd_ell}
 The average D2D download time for the $\ell$th request, $\TDD^{(\ell)}$, is independent of the specific request if the index $\ell$ is sufficiently large. 
 \end{Corollary}
\begin{proof}
Similarly to the average D2D download delay, $\TDD^{(\ell)}$ is
\begin{equation}
\TDD^{(\ell)} = \TDDs^{(\ell)} \Pr\{ R^{(\ell)}=1\} + \TDDr^{(\ell)}\Pr\{ R^{(\ell)}=0\}\, , \nonumber \\  
\end{equation}
where 
\begin{align}
\TDDi^{(\ell)}=& \tn \Big( (k-i) \Pr\{ \Sv^{(\ell)}_{[k-i]}=\onev_{k-i} | R^{(\ell)}=i\} +  c_{k,i}\Pr \{\ S^{(\ell)}_{1}=0 | R^{(\ell)}=i\} \nonumber\\
		      &+ \sum_{j=0}^{k-1-i} (j+1)  \Pr \{\Sv^{(\ell)}_{[j]}=\onev_{j}, S^{(\ell)}_{j+1}=0 | R^{(\ell)}=i\} \Big)\, , \qquad \qquad i=0,1 \nonumber 
\end{align}
The Lemma follows from the fact that the probabilities in the expressions above are independent of $\ell$, when $\ell$ grows large.  
\end{proof}
\section{Numerical Results}
\label{sec:NumericalResults}


%

In this section, we evaluate the average download delay when content is cached using MDS codes for a cluster with $\Mc=30$ nodes on average, departure rate $\mu=1$, and request rate $\omega=0.02$. We compare the average file download delay $\TDW$ of the considered network with MDS-coded DS with the delay of the traditional scenario where the content is solely downloaded from the BS, denoted by $T_\text{ref}=k \tbs$ and with uncoded caching.  In the following, with no loss of generality, we set $T_\text{ref}=1$ t.u.. 

We recall that each file that is cached in the DS library is divided into $k$ symbols, and encoded using an ($n,k$) MDS code, where $n$ is the number of storage nodes in the cell. The code parameters are chosen such that $k\leq \nc$. In this way, the \emph{average storage overhead} in a cluster, $\nc-k$, is positive, which increases the probability that the content is downloaded through D2D communication only. Alternatively, we can use the same $(\nc,k)$ MDS code for each cluster, but in this case the BS must continuously restore the initial state of reliability of the DS network when storage nodes leave the clusters~\cite{Ped16}. 

We first consider the special case where $F=Z$, therefore the probability of hitting the cache $\Pr\{H=1\}$ is $1$. 
In Figs.~\ref{f:0.1}--\ref{f:0.001}, we show the gain that can be achieved using MDS-coded caching, by reporting the ratio between $T_\text{ref}$ and $\TDW$ as a function of the update interval $\Delta$. The  infinite series involved in the computation of $\TDW$ are truncated to a given value $t$, chosen according to $\text{argmin}_{t>\nc}\{ \pi_t(\nc)<10^{-5} \}$ when involving the number of storage nodes and to $\text{argmin}_{t>\Mc}\{ \pi_t(\Mc)<10^{-5} \}$ when involving the number of nodes in general. We fix the ratio $k/\nc$ to be $1/3$, and consider several MDS codes. Moreover, we also consider an uncoded scenario where one storage node on average in the cluster stores the uncoded files. In Figs.~\ref{f:0.1}, \ref{f:0.01}, and \ref{f:0.001}, $\tn$ is $10$, $100$, and $1000$ times, respectively, smaller than $\tbs$. In the figures, the solid lines correspond to the analytical closed-form expressions derived in the previous sections and markers correspond to simulation results. It is observed that the analytical expressions predict very well the actual performance, which shows the goodness of the approximations introduced in~(\ref{e:approxI}) and (\ref{e:pq_approx}).
The results clearly show that MDS-coded DS can greatly improve the performance in terms of content download delay with respect to the case where content is downloaded from the BS, provided that the update interval, $\Delta$, is sufficiently small. For example, for $\tbs=10 \tn$ and $\Delta=1$, a speed-up factor of around $19$ in the download is achieved with respect to the case of downloading from the BS using a $(15,5)$ MDS code. Interestingly, the results also show that the performance improves when $k$ increases. In particular, simple replication (repetition coding) is very inefficient and much better performance are achieved using larger MDS codes (of the same rate).

We now consider the more general case where only part of the library of files is cached in the devices. We assume that the library has a size of $Z=1000$ files, and each storage node stores one symbol for each of the $F$ most popular files. We assume that the each file is of size $100$ MB, which corresponds to a $10$-minutes video. We further assume that each storage node makes $6$ GB available for caching. In this case, the number of cached files $F$, that corresponds to the number of symbols cached by each storage node, increases by increasing $k$, since the size of one encoded symbols is $100$MB$/k$ and $F$ is clearly obtained by dividing the storage capacity of the node by the symbol size. In Fig.~\ref{f:sigma}, we show the download speedup factor $T_\text{ref}/\TDW$ as a function of the parameter $\sigma$, which regulates the relative popularity of the files. For example, a large value of $\sigma$ represents the case where few popular files are responsible for the majority of the download traffic. The figure refers to the case where $\tbs=100\tn$ and $\Delta=0.5$. The results confirm the gain that can be achieved by MDS-coded distributed caching, but highlight another important aspect. When the whole library is not cached in the network, i.e., $F<Z$, the gain reduction is not negligible, especially for low values of $\sigma$. This fact  suggests that the adopted deterministic allocation, where the $F$ most popular files are equally stored in the DS network, is suboptimal.

\begin{figure}[!t]
\centering{} \includegraphics[width=\columnwidth]{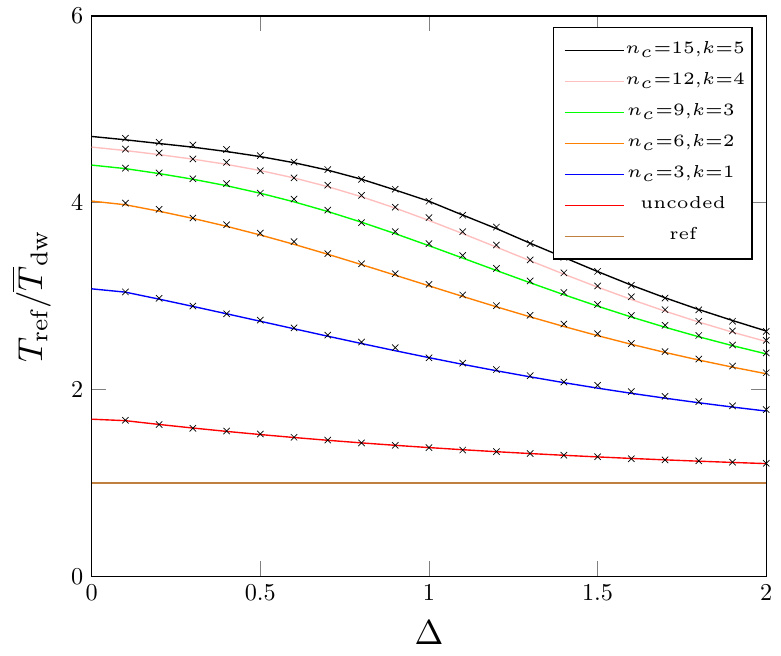}
\vspace{-4mm}
 \caption{Ratio between the file download delay without D2D communication and that of the scenario using MDS-coded distributed caching. $\tbs=10 \tn$. Solid lines show analytical results and markers simulation results.}\label{f:0.1}
 \vspace{-3mm}
\end{figure}

\begin{figure}[!t]
\centering{} \includegraphics[width=\columnwidth]{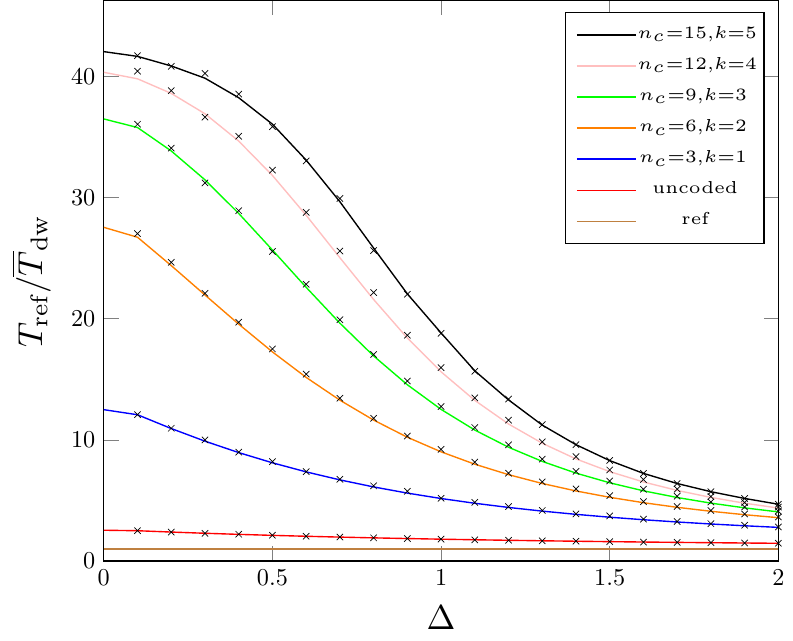}
\vspace{-4mm}
 \caption{Ratio between the file download delay without D2D communication and that of the scenario using MDS-coded distributed caching. $\tbs=100 \tn$. Solid lines show analytical results and markers simulation results.}\label{f:0.01}
  \vspace{-3mm}
\end{figure}

\begin{figure}[!t]
\centering{} \includegraphics[width=\columnwidth]{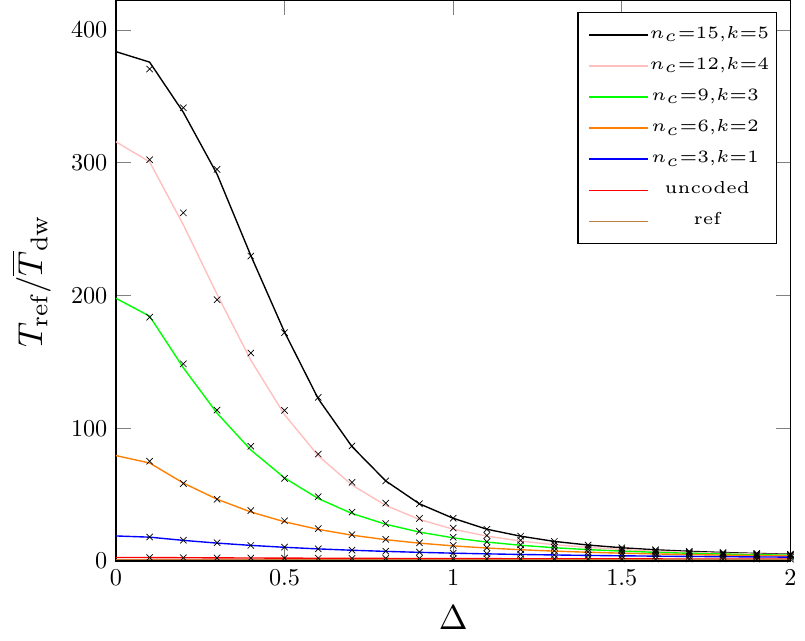}
\vspace{-4mm}
  \caption{Ratio between the file download delay without D2D communication and that of the scenario using MDS-coded distributed caching. $\tbs=1000 \tn$. Solid lines show analytical results and markers simulation results.}\label{f:0.001}
   \vspace{-5mm}
\end{figure}
\begin{figure}
\centering{} \includegraphics[width=\columnwidth]{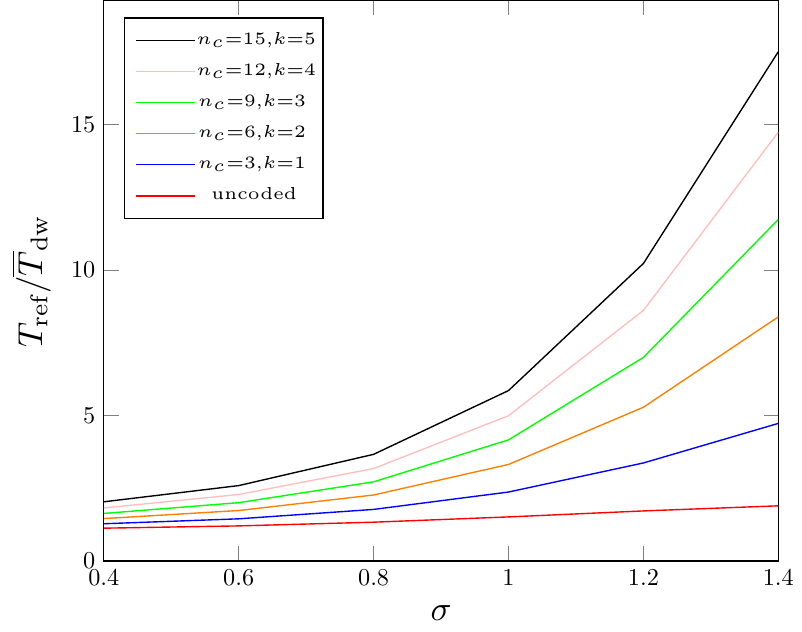}
\vspace{-4mm}
  \caption{Ratio between the file download delay without D2D communication and that of the scenario using MDS-coded distributed caching for $F<Z$. $\tbs=100 \tn$ and $\Delta=0.5$.}\label{f:sigma}
   \vspace{-5mm}
\end{figure}

\section{Conclusions}
\label{sec:Conclusions}

In this paper, we considered the cache of popular content in the mobile devices of a cellular network using maximum distance separable erasure correcting codes to speed-up content delivery. We derived analytical expressions for the average download delay and showed that MDS-coded distributed caching may dramatically reduce the download delay with respect to the traditional case where content is always downloaded from the base station.



\appendices

\section{Proof of Lemma 1}
\label{app:ProofLem1}
We denote by $X_1^{(\ell)}$ the number of storage nodes available for download at the time of the $\ell$th download request, i.e., the number of storage nodes of the DS list that have not left the cluster at the time of the request. We compute $\Pr\{X_1=x\}$ by averaging over an infinite number of requests, 
\begin{equation}\label{e:x1}
\Pr\{X_1=x \}=\lim_{L \to \infty} \frac{1}{L} \sum_{\ell=1}^L \Pr\{X_1^{(\ell)}=x \}\, .
\end{equation}
Similarly, let $Y^{(\ell)}$ be the number of storage nodes at the beginning of the update interval wherein the $\ell$th request arrives, denoted by $\Delta^{(\ell)}$. We have
\begin{equation}\label{e:x1l}
\Pr\{X_1^{(\ell)}=x\}=\sum_{y=0}^\infty \Pr\{X_1^{(\ell)}=x | Y^{(\ell)}=y\} \Pr\{Y^{(\ell)}=y\}\,.
\end{equation}
In~\cite{Ped16}, it was shown that the probability  $\Pr\{X_1^{(\ell)}=x | Y^{(\ell)}=y\} $ does not depend on $\ell$ (when $\ell$ grows large), and is given by~(\ref{e:px_J}). Its derivation is based on the observation that the number of storage nodes available for download in the update interval is described by a Poisson death process. The probability $ \Pr\{Y^{(\ell)}=y\}$ can be written as
\begin{equation}\label{e:y}
 \Pr\{Y^{(\ell)}=y\}= \Pr\{ \tilde{Y}^{(\ell)} =y | \text{req. in }\Delta^{(\ell)} \}=\frac{ \Pr\{ \text{req. in }\Delta^{(\ell)}  | \tilde{Y}^{(\ell)} =y \}\Pr\{ \tilde{Y}^{(\ell)} =y \} } {  \Pr\{ \text{req. in }\Delta^{(\ell)}  \} }\, ,
\end{equation}
where in the second equality we used Bayes' rule. In~(\ref{e:y}), $\tilde{Y}^{(\ell)}$ is the number of storage nodes at the beginning of the update interval, which is described by a birth-death Poisson process, thus $\Pr\{ \tilde{Y}^{(\ell)} =y \} =\pi_y(\nc)$.
The probability $\Pr\{ \text{req. in }\Delta^{(\ell)}   \}$ is the probability that there is at least one request in $\Delta^{(\ell)}$.  It depends on the inter-request time, which in turn depends on the number of nodes in the cluster. Therefore, we compute
\begin{equation}\nonumber
\Pr\{ \text{req. in }\Delta^{(\ell)}   \}= \sum_{m=1}^\infty (1-e^{-m\omega \Delta}) \pi_m(\Mc) \, ,
\end{equation}
where $1-e^{-m\omega \Delta}$ is the probability that the inter-request time is shorter than $\Delta$ when $m$ nodes are present in the cluster.
Similarly, we compute $\Pr\{ \text{req. in }\Delta^{(\ell)}  | \tilde{Y}^{(\ell)} =y \}$ as
\begin{equation}\nonumber
\Pr\{ \text{req. in }\Delta^{(\ell)}  | \tilde{Y}^{(\ell)} =y \} =\sum_{m=y}^\infty (1-e^{-m\omega \Delta}) \pi_{m-y}(\Mc-\nc)\,,
\end{equation}
where $\pi_{m-y}(\Mc-\nc)$ is the probability that there are $m$ nodes in the cluster, given that there are $y$ storage nodes. Since these probabilities are independent of the specific request, we conclude that $ \Pr\{Y^{(\ell)}=y\}$ is also independent of $\ell$. Substituting (\ref{e:y}) into (\ref{e:x1l}), we observe that $\Pr\{X_1^{(\ell)}=x\}$ is also independent of $\ell$ and using~(\ref{e:x1}) we prove the lemma.

\section{Proof of Lemma 4}
\label{app:ProofLem4}

We compute the conditional probability that no symbols are downloaded by averaging over an infinite number of requests, 
\begin{equation}\nonumber
\Pr\{ S_1=0|R=i\} =\lim_{L \to \infty} \frac{1}{L} \sum_{\ell=1}^L \Pr\{ S_1^{(\ell)}=0  |R^{(\ell)}=i\}\, ,
\end{equation}
where $S_1^{(\ell)}\in\{0,1\}$ is the RV describing the download of the first symbol for the $\ell$th request. We now consider the computation of $\Pr\{ S_1^{(\ell)}=0|R^{(\ell)}=i\}$. The recovery of the first symbol fails with probability $1$ if the requesting node leaves the cell before completing the download. It also fails if the requesting node stays in the cluster but no storage nodes are available or if it chooses to download from a storage node which departs before $\tn$ t.u. from the start of the download. Let $O^{(\ell)}$ be the departure time of the node which places the $\ell$th request and let $\mathcal{A}_1^{(\ell)}=\{  O^{(\ell)} -  W^{(\ell)} > \tn\}$ be the event that the node which places the $\ell$th request stays in the network for more than $  \tn $ t.u. from the start of the download. The corresponding probability does not depend on $\ell$ and is easily computed as $\Pr\{ \mathcal{A}_1^{(\ell)}\} =e^{-\mu \tn}$. Similarly, the probability that the requesting node departs before $\tn$ t.u. from the start of the download is $(1-  e^{-\mu \tn} )$. Therefore, the conditional probability that the $\ell$th request fails the first symbol download is 
\begin{equation}\label{e:s1l}
\Pr\{ S_1^{(\ell)}=0|R^{(\ell)}=i\}=(1-  e^{-\mu \tn} )+   e^{-\mu \tn} \Pr\{ S_1^{(\ell)}=0  | \mathcal{A}_1^{(\ell)}, R^{(\ell)}=i\}\, .
\end{equation}
Let $G^{(\ell)}_1\in\{0,\dots,\infty \}$ be the number of storage nodes useful for download for the $\ell$th request, i.e., the number of storage nodes in the DS list that have not left the cluster at the time of the $\ell$th request, excluding the requesting node itself if it belongs to the DS list. Let $D^{(\ell)}_1\in\{0,\dots,\infty\}$ the number of departures in $\tn$ t.u. among the $G^{(\ell)}_1$ storage nodes. The probability $ \Pr\{ S_1^{(\ell)}=0  | \mathcal{A}_1^{(\ell)}, R^{(\ell)}=i\} $ can be written as
\begin{align}
 &\Pr\{ S_1^{(\ell)}=0  | \mathcal{A}_1^{(\ell)}, R^{(\ell)}=i\} =\nonumber\\
 &=\sum_g\sum_d  \Pr\{ S_1^{(\ell)}=0  | \mathcal{A}_1^{(\ell)}, R^{(\ell)}=i, G_1^{(\ell)}=g, D_1^{(\ell)}=d  \} \Pr\{ G_1^{(\ell)}=g, D_1^{(\ell)}=d| \mathcal{A}_1^{(\ell)}, R^{(\ell)}=i \}\nonumber \\ 
&= \sum_g\sum_d  \Pr\{ S_1^{(\ell)}=0  | \mathcal{A}_1^{(\ell)}, G_1^{(\ell)}=g, D_1^{(\ell)}=d  \} \Pr\{  G_1^{(\ell)}=g, D_1^{(\ell)}=d | R^{(\ell)}=i \}\, . \label{e:s1cond}
\end{align}
Equation~(\ref{e:s1cond}) is obtained by observing that $i)$ the probability that the download of the first symbol fails conditioned to $G_1^{(\ell)}$ and $D_1^{(\ell)}$ is independent of the type of request, and that $ii)$ the number of useful storage nodes and the number of departures in $\tn$ t.u. among them is independent of the departure of the requesting node.  The probability   $\Pr\{ S_1^{(\ell)}=0  | \mathcal{A}_1^{(\ell)}, G_1^{(\ell)}=g, D_1^{(\ell)}=d  \}$ is equal to 1 if there are no useful storage nodes, i.e., $g=0$. Otherwise, it equals the probability to choose one of the $d$ storage nodes that leave the cell in $\tn$ t.u., i.e., $d/g$, with $d\leq g$. We observe that the number of departures of useful storage nodes conditioned to their number is independent of the type of request and that the number of storage nodes useful for download $G_1^{(\ell)}$ is related to $X^{(\ell)}_1$ by
\begin{equation}\nonumber
\Pr\{G_1^{(\ell)}=g|R^{(\ell)}=i\}=\Pr\{ X_1^{(\ell)}=g+i |R^{(\ell)}=i\}\, ,
\end{equation}
since when the request originates from a storage node of the DS list, the requesting node itself is not counted among the useful storage nodes. Therefore, the probability $\Pr\{  G_1^{(\ell)}=g, D_1^{(\ell)}=d | R^{(\ell)}=i \}$ can be written as
\begin{equation}
\Pr\{  G_1^{(\ell)}=g, D_1^{(\ell)}=d | R^{(\ell)}=i \}= \Pr\{  D_1^{(\ell)}=d|G_1^{(\ell)}=g\} \Pr\{X_1^{(\ell)}=g+i|R^{(\ell)}=i  \}\, .
\end{equation}
We denote by $\theta(d,g)$ the probability $\Pr\{  D_1^{(\ell)}=d|G_1^{(\ell)}=g\}$, given in~(\ref{e:theta}). Its derivation is similar to that of $\Pr\{X_1 | Y \}$~\cite{Ped16}. The probability $\Pr\{X_1^{(\ell)}=g+i|R^{(\ell)}=i  \}$ is independent of the specific request and is given by~(\ref{e:px1cond}).

 After simple manipulations, we finally obtain
\begin{equation}\label{e:PS1}
\Pr\{ S_1^{(\ell)}=0|R^{(\ell)}=i\}=   1 + e^{-\mu \tn}  \Big(  \Pr\{X_1=i|R=i\} + \sum_{g =1}^\infty \sum_{d =0 }^{g} \frac{d}{g}    \Pr\{X_1=g+i|R=i\} \theta(d,g) - 1\Big).
\end{equation}
Since the probabilities involved in~(\ref{e:PS1}) are all independent of $\ell$, the lemma is proved.

\section{Proof of Lemma 5}
\label{app:ProofLem5}

To evaluate the probability of complete download from the DS network, we start with the following limit,
\begin{equation}\nonumber
\Pr\{  \Sv_{[k-i]} =\onev_{k-i}| R=i \}= \lim_{L \to \infty}   \frac{1}{L} \sum_{\ell=1}^L \Pr\{  \Sv_{[k-i]}^{(\ell)}=\onev_{k-i}|R^{(\ell)}=i \}\, ,
\end{equation}
where $\Sv_{[j]}^{(\ell)}=(  S_1^{(\ell)}, \dots, S_j^{(\ell)})$ and $S_i^{(\ell)}$ describes the successful symbol download at the $i$th attempt of the $\ell$th request. We consider the $\ell$th request and, similarly to the proof of Lemma~\ref{l:S1}, we will find that this probability is independent of $\ell$.  We denote by the RV $G_j^{(\ell)}\in\{0,\ldots,\infty\}$ the number of storage nodes useful for download at the time of the $j$th attempt of the $\ell$th request, i.e., the storage nodes of the DS list not yet contacted, excluding the requesting node if it belongs to the DS list. We denote by the RV $D_j^{(\ell)}  \in \{0,\ldots,\infty\}$ the number of departures in $\tn$ t.u. among the $G_j^{(\ell)}$ nodes. We also denote by $\mathcal{A}_j^{(\ell)}=\{  O^{(\ell)} -  W^{(\ell)} > j \tn\}$ the event that the node which places the $\ell$th request stays in the network for more than $ j \tn $ t.u. from the start of the download. The corresponding probability does not depend on $\ell$ and is given by $\Pr\{ \mathcal{A}_j^{(\ell)}\} =e^{-j\mu \tn}$. The probability of complete download is zero if the requesting node departs before $ (k-i)\tn $ t.u. from the start of the download, therefore we can write
\begin{equation}
\Pr\{  \Sv_{[k-i]}^{(\ell)}=\onev_{k-i}|R^{(\ell)}=i \}=e^{-(k-i)\mu \tn} \Pr\{  \Sv_{[k-i]}^{(\ell)}=\onev_{k-i}|\mathcal{A}_{k-i}^{(\ell)}, R^{(\ell)}=i  \}\,.
\end{equation}
The probability $\Pr\{  \Sv_{[k-i]}^{(\ell)}=\onev_{k-i}|\mathcal{A}_{k-i}^{(\ell)}, R^{(\ell)}=i  \}$ can be written as
\begin{align}
&\Pr\{  \Sv_{[k-i]}^{(\ell)}=\onev_{k-i}|\mathcal{A}_{k-i}^{(\ell)}, R^{(\ell)}=i  \}=\nonumber\\
&=\sum_{g} \sum_{d} \Pr\{S^{(\ell)}_{k-i}=1|  \mathcal{A}_{k-i}^{(\ell)}, G_{k-i}^{(\ell)}=g, D_{k-i}^{(\ell)}=d  \} \Pr\{  G_{k-i}^{(\ell)}=g, D_{k-i}^{(\ell)}=d,\Sv_{[k-i-1]}^{(\ell)}=\onev_{k-i-1} |  \mathcal{A}_{k-i}^{(\ell)}, R^{(\ell)}=i \}\nonumber \\
&=\sum_{g=1}^\infty \sum_{d=0}^{g} \frac{g-d}{g} \gamma^{(\ell)}_{k-i}(g,d,i)\, .  \label{e:gamma1} 
\end{align}
The last equality is obtained by observing that the probability $\Pr\{ S_{k-i}^{(\ell)}=1  |  \mathcal{A}_{k-i}^{(\ell)},  G_{k-i}^{(\ell)}=g,  D_{k-i}^{(\ell)}=d\}$ for $g>0$ and $d<g$ equals the probability to choose one of the storage nodes of the DS list that remains in the cluster, i.e., $\frac{g-d}{g}$. Moreover, we have defined $\gamma^{(\ell)}_j(g,d,i)\triangleq\Pr\{  G_{j}^{(\ell)}=g, D_{j}^{(\ell)}=d, \Sv_{[j-1]}^{(\ell)}=\onev_{j-1} |  \mathcal{A}_{j}^{(\ell)}, R^{(\ell)}=i \}$, that can be computed as
\begin{align}
&\gamma^{(\ell)}_j(g,d,i)=\Pr\{D_{j}^{(\ell)}=d| G_{j}^{(\ell)}=g\} \sum_{g'} \sum_{d'}  \Pr\{G_{j}^{(\ell)}=g|G_{j-1}^{(\ell)}=g', D_{j-1}^{(\ell)}=d', S_{j-1}^{(\ell)}=1 \}  \nonumber \\
&\quad\cdot \Pr\{S^{(\ell)}_{j-1}=1|  \mathcal{A}_{j}^{(\ell)}, G_{j-1}^{(\ell)}=g', D_{j-1}^{(\ell)}=d' \} \Pr\{  G_{j-1}^{(\ell)}=g', D_{j-1}^{(\ell)}=d', \Sv_{[j-2]}^{(\ell)}=\onev_{j-2} |  \mathcal{A}_{j-1}^{(\ell)}, R^{(\ell)}=i \}\, ,\nonumber
\end{align}
for $j>1$. We also define $ E(g,g',d')\triangleq \Pr\{G_{j}^{(\ell)}=g|G_{j-1}^{(\ell)}=g', D_{j-1}^{(\ell)}=d', S_{j-1}^{(\ell)}=1\}$, which is equal to one if $g=g'-d'-1$ and $g'>d'$, and zero otherwise. The condition $g=g'-d'-1$ follows from the fact that the number of useful storage nodes after a successful symbol download is equal to the number of useful storage nodes still alive, \mbox{$g'-d'$}, minus the storage node just used. The condition $g'>d'$ comes from the fact that the $(j-1)$th symbol download is assumed to be successful, i.e., $S_{j-1}^{(\ell)}=1$.   

It is easy to prove by induction that the probability in~(\ref{e:gamma1}) does not depend on $\ell$. By defining $\gamma_j(g,d,i)\triangleq \gamma^{(\ell)}_j(g,d,i)$, we obtain the recursion~(\ref{e:gamma}), with initial condition~(\ref{e:gammaIC}). The probabilities $\gamma_j(g,d,i)$, $j\geq 1$, are equal to zero for $d>g$.

\bibliographystyle{IEEEtran}

\end{document}